\documentclass[11pt]{article} % use larger type; default would be 10pt

\usepackage[utf8]{inputenc} % set input encoding (not needed with XeLaTeX)

\usepackage{geometry} % to change the page dimensions
\usepackage{graphicx,psfrag} % support the \includegraphics command and options
\usepackage{color}
\usepackage{enumerate}

\definecolor{Red}{rgb}{1,0,0}
\definecolor{Blue}{rgb}{0,0,1}
\definecolor{Olive}{rgb}{0.41,0.55,0.13}
\definecolor{Green}{rgb}{0,1,0}
\definecolor{MGreen}{rgb}{0,0.8,0}
\definecolor{DGreen}{rgb}{0,0.55,0}
\definecolor{Yellow}{rgb}{1,1,0}
\definecolor{Cyan}{rgb}{0,1,1}
\definecolor{Magenta}{rgb}{1,0,1}
\definecolor{Orange}{rgb}{1,.5,0}
\definecolor{Violet}{rgb}{.5,0,.5}
\definecolor{Purple}{rgb}{.75,0,.25}
\definecolor{Brown}{rgb}{.75,.5,.25}
\definecolor{Grey}{rgb}{.5,.5,.5}

\usepackage{booktabs} % for much better looking tables
\usepackage{array} % for better arrays (eg matrices) in maths
\usepackage{paralist} % very flexible & customisable lists (eg. enumerate/itemize, etc.)
\usepackage{verbatim} % adds environment for commenting out blocks of text & for better verbatim
\usepackage{subfigure} % make it possible to include more than one captioned figure/table in a single float
\usepackage{amsmath,amssymb,amsthm}
\usepackage{fancyhdr} % This should be set AFTER setting up the page geometry
\usepackage{sectsty}
\allsectionsfont{\sffamily\mdseries\upshape} % (See the fntguide.pdf for font help)
\theoremstyle{plain}

\newtheorem{theorem}{Theorem}

\newtheorem{corollary}{Corollary}
\newtheorem{claim}{Claim}
\newtheorem{lemma}{Lemma}

\theoremstyle{remark}

\newtheorem{remark}{Remark}
\newtheorem{observation}{Observation}

\theoremstyle{definition}
\newtheorem{definition}{Definition}

\newcommand{\p}{{\rm P}}

\def\cC{{\cal C}}

\def\cU{{\cal U}}

\def\cW{{\cal W}}
\def\cX{{\cal X}}
\def\cY{{\cal Y}}

\def\d{{\rm d}}

\renewcommand\appendix{\par
\setcounter{section}{0}%
\setcounter{subsection}{0}%
\setcounter{table}{0}
\setcounter{figure}{0}
\gdef\thetable{\Alph{table}}
\gdef\thefigure{\Alph{figure}}
\section*{Appendix}
\gdef\thesection{\Alph{section}}
\setcounter{section}{1}}

\title{On broadcast channels with binary inputs and symmetric outputs}
\author{Yanlin Geng \and Chandra Nair \and Shlomo Shamai\thanks{The work of S. Shamai was supported by the Israel Science Foundation (ISF).} \and Zizhou Vincent Wang}

\begin{document}
\maketitle

%%%%%%%%%%%%%%%%%%%%%%%%%%%%%%%%%%%%%%%%%%%%%%%%%%%%%%%%%%%%%%%%%%%%%%%%%%

\begin{abstract}
We study the capacity regions of broadcast channels with binary inputs and symmetric outputs. We study the partial order induced by the more capable ordering of broadcast channels for channels belonging to this class.  This study leads to some surprising connections regarding various notions of dominance of receivers. The results here also help us isolate some classes of symmetric channels where the best known inner and outer bounds differ.
\end{abstract}
\section{Introduction}

In \cite{cov72}, Cover introduced the notion of a broadcast channel through
which one sender transmits information to two or more receivers. For the purpose
of this paper we focus our attention on broadcast channels with precisely two
receivers.

{\em Definition: } A {\em broadcast channel} (BC) consists of an
input alphabet $\mathcal{X}$ and output alphabets $\mathcal{Y}_1$
and $\mathcal{Y}_2$ and a probability transition function
$p(y_1,y_2|x)$. A $((2^{nR_1}, 2^{nR_2}),n)$ code for a broadcast
channel consists of an encoder
\[ x^n : 2^{nR_1} \times 2^{nR_2} \rightarrow \mathcal{X}^n, \]
and two decoders
\[ \hat{\cW}_1 : \mathcal{Y}_1^n \rightarrow 2^{nR_1} \]
\[ \hat{\cW}_2 : \mathcal{Y}_2^n \rightarrow 2^{nR_2}. \]

The probability of error $P_e^{(n)}$ is defined to be the
probability that the decoded message is not equal to the transmitted
message, i.e., 

\[ P_e^{(n)} = \mathbf{P}\left(\{\hat{\cW}_1(Y_1^n) \neq \cW_1 \}\cup \{
\hat{\cW}_2(Y_2^n) \neq \cW_2 \} \right)
\]
where the message is assumed to be uniformly distributed over
$2^{nR_1} \times 2^{nR_2}$. 

A rate pair $(R_1,R_2)$ is said to be {\em achievable } for the
broadcast channel if there exists a sequence of $((2^{nR_1},
2^{nR_2}),n)$ codes with $P_e^{(n)} \rightarrow 0$. The {\em
capacity region} of the broadcast channel is the closure of the set of
achievable rates. {\em The capacity region of the two-receiver discrete
memoryless channel is unknown.}

The capacity region is known for lots of special cases where there is a
``dominant receiver'' such as degraded, less noisy, more capable, essentially
less noisy, and essentially more capable. In fact superposition coding is
optimal here. An interesting observation in \cite{nai08b} was that the notions
of more capable and essentially less noisy may not be compatible with each
other.

In this paper, we study in detail the notions of more capable receivers and
essentially less noisy receivers by focusing on an important(commonly used in
coding theory) class of binary-input symmetric-output(BISO) broadcast channels.
We establish a slew of results and some of the interesting ones are summarized
below.

\subsection{Summary of selected results}
Here are some of the results established in this paper. 
\begin{itemize}
 \item Any BISO channel with capacity $C$ is more capable than the binary symmetric channel with capacity $C$. (Corollary \ref{cor:bsc})
 \item The binary erasure channel with capacity $C$ is more capable than any BISO channel with capacity $C$. (Corollary \ref{cor:bec})
 \item Any two BISO channels with the same capacity and whose outputs have cardinality at most 3, are more-capable comparable, i.e. one receiver is more capable than the other receiver. (Corollary \ref{thm:card-3})
 \item For any two BISO channels with same capacity, a receiver $Y_1$ is more capable than receiver $Y_2$ {\em if and only if} receiver $Y_2$ is  essentially less noisy  than $Y_1$. (They go in reverse directions !) (Lemma \ref{le:eq})
 \item Superposition coding region is the capacity region for a BISO-broadcast channel if any one of the channels is either a BSC or a BEC. (Corollary \ref{cor:sup})
\item For two BISO channels with the same capacity, superposition coding is optimal if and only if the channels are more capable comparable. (Corollary \ref{thm:BISO-n-comp})
\item For two BISO channels of same capacity Marton's inner bound differs from the outer bound\cite{nae07} unless the channels are more capable comparable (Theorem \ref{thm:TD-RTD-OB})
\item We also show that it suffices to consider $U \to X$ to be BSC when we wish to compute the boundary of the superposition coding region for BISO broadcast channels. (Lemma \ref{lem:BSC-aux-deg}). This vastly generalizes a result of Wyner and Ziv\cite{wyz73} for degraded BSC broadcast channel.
\end{itemize}

\subsection{Preliminaries}

\begin{definition}\cite{kom75}
A channel $F_1: X \to Y_1$ is said to be {\em more capable} than the channel
$F_2: X \to Y_2$, denoted $F_1 \gg F_2$, if $I(X;Y_1) \geq I(X;Y_2), \forall
p(x)$.
\end{definition}

\begin{definition}\cite{nai08b}
A class of distributions $\mathcal{P} = \{p(x)\}$ on the input alphabet
$\mathcal{X}$ is said to be a {\em sufficient class} of distributions for a
2-receiver broadcast channel if the following holds:
Given any triple of random variables $(U,V,X)$ satisfying $(U,V) \to X \to
(Y_1,Y_2)$ forms a Markov chain, there exists a distribution $q(u,v,x)$ (also
obeying the Markov relationship $(U,V) \to X \to (Y_1,Y_2)$) that satisfies
\begin{align}
& q(x) \in \mathcal{P}, \nonumber \\
&I(U;Y_i)_p \leq I(U;Y_i)_q,~ i=1,2, \nonumber \\
&I(V;Y_i)_p \leq I(V;Y_i)_q, ~i=1,2, \nonumber \\
&I(X;Y_i|U)_p \leq I(X;Y_i|U)_q, ~i=1,2, \label{eq:suff}\\
&I(X;Y_i|V)_p \leq I(X;Y_i|V)_q, ~i=1,2, \nonumber \\
&I(X;Y_i)_p \leq I(X;Y_i)_q, ~i=1,2,\nonumber
\end{align}
\end{definition}

\begin{definition}\cite{nai08b}
A channel $F_1: X \to Y_1$ is {\em essentially less noisy} compared to a channel
$F_2: X \to Y_2$, denoted by $F_1 \succeq F_2$, if there exists a sufficient
class of distributions $\mathcal{P}$ such that whenever $p(x) \in \mathcal{P}$,
for all $U \to X \to (Y_1,Y_2)$ we have
$$ I(U;Y_2) \leq I(U;Y_1). $$
\end{definition}

In this paper, we restrict ourselves to a class $\cC$  of discrete memoryless
channels with binary inputs and symmetric outputs(BISO) as defined below.

\begin{definition}
A discrete memoryless channel with input alphabet $\cX=\{0,1\}$ and output
alphabet $\cY=\{k: -l \leq k \leq l\}$ is said to belong to class $\cC$ (or
BISO) if
$$ p_k = \p(Y=k|X=0) = \p(Y=-k|X=1), -l \leq k \leq l. $$
\end{definition}
Binary symmetric channel(BSC) and Binary Erasure Channel(BEC) are examples of
channels that belong to the class $\cC$. It is easy to see that uniform input
distribution is the capacity achieving distribution for any channel in $\cC$.

\begin{remark}
As $k=0$ can be split equally into $0^+$ and $0^-$ with probability
$p_{0^+}=p_{0^-}=p_0/2$, so we just consider $k=\pm1,...,\pm l$ and use
$\{p_k,p_{-k}:k=1,\ldots,l\}$ to denote the transition probabilities. Sometimes
shortened to $\{p_k,p_{-k}\}_k$.
\end{remark}

Partition $P$ of an interval $[a, b]$ is a finite sequence (points) $\{t_k\}_k$
such that $a =t_0 < t_1 < t_2 < \ldots < t_N = b$. A partition $P$ is finer than
$Q$ if points of partition  $P$   contain those of $Q$. A
common refinement of two partitions $P$ and $Q$ is a new partition consisting of
all the points of $P$ and $Q$.

\begin{definition} \label{def:BISO-parti-curve}
(BISO partition and BISO curve)\\
For a BISO channel with transition probabilities $\{p_k,p_{-k}\}_k$, rearrange
$h(\frac{p_k}{p_k+p_{-k}})$ in the ascending order and denote the permutation as
$\pi$. {\em BISO partition} is defined as the partition of $[0,1]$ with points
$t_k =\sum_{i=1}^k (p_{\pi_i}+p_{-\pi_i})$. We set $t_0 = 0$. {\em BISO curve}
is defined as the stepwise function $f(t)$ such that
$f(t)=h(\frac{p_{\pi_k}}{p_{\pi_k}+p_{-\pi_k}})$ on $(t_{k-1},t_k]$, and
$f(0)=0$.
\end{definition}

For the channel $BSC(p)$, we have the partition as $t_0=0, t_1=1$ and the curve
as $f(t)=h(p)$ on $(0,1]$. For the channel $BEC(e)$, we have the partition as
$t_0=0, t_1=1-e, t_2=1$, and the curve as $f(t) = 0$ on $(0,1-e]$ and $f(t)=1$
on $(1-e,1]$.

\begin{definition} \label{def:Lorenz-curve}
(Lorenz curve of a BISO channel)\\
For a BISO channel with BISO curve $f(t)$, the Lorenz curve (or the cumulative
function) $F(t)$ is defined as $F(t)=\int_0^t f(\tau)\d \tau$.
\end{definition}

\noindent{\em Properties of the Lorenz curve:}

Since $0 \leq f(t) \leq 1$ and $f(t)$ is non-decreasing on $[0,1]$ we have
\begin{enumerate}
\item $F(t)$ is non-negative, piecewise linear and convex.
\item The slope of the line segments of $F(t)$ is at most 1.
\end{enumerate}

By definition of BISO curve, the length of $k$-th interval $(t_{k-1},t_k]$ is
$(p_{\pi_k}+p_{-\pi_k})$. Therefore
\begin{align}
I(X;Y)=&\sum_{k>0}{(p_k+p_{-k})h(x*h^{-1}(h(\frac{p_k}{p_k+p_{-k}})))} -
\sum_{k>0}{(p_k+p_{-k})h(\frac{p_k}{p_k+p_{-k}})} \label{eq:ixy}\\
=& \int_0^1 h(x*h^{-1}(f(\tau)))\d\tau - \int_0^1 f(\tau)\d\tau \nonumber \\
= & \int_0^1 h(x*h^{-1}(f(\tau)))\d\tau - F(1)  \nonumber
\end{align}
Thus, a finer partition does not change $I(X;Y)$ and in particular the channel
capacity. Indeed the capacity is $C=1-F(1)$.

\section{Main}

\subsection{On partial orderings and capacity regions of BISO broadcast channels}

\subsubsection{On more capable comparability of BISO channels}

We will establish a sufficient condition for determining whether two BISO
channels are comparable using the more capable partial ordering. Before we state
our sufficient condition for more capable comparable, we need the following
three lemmas.

\begin{lemma}\label{lem:cont-major}
Given BISO channels $X\to Y$ and $X\to Z$ with BISO curves $f(t)$ and
$g(t)$, respectively. Let the common refinement of these two BISO partitions be
$\{t_k:k=0,\ldots,\hat{N}\}$, and $\xi_k=t_k-t_{k-1}$. Then $$F(t_i)=
\sum_{k=1}^{i} \xi_k f(t_k) \leq \sum_{k=1}^{i} \xi_k g(t_k) = G(t_i),\quad
i=1,\ldots, \hat N$$  if and only if the Lorenz curve $F(t)
\leq G(t)$ for all $t\in[0,1]$.
\end{lemma}
\begin{proof}
The {\em if} direction is obvious. We just need to prove the other direction,
i.e. $F(t_i) \leq G(t_i) \Rightarrow F(t) \leq G(t) $. We prove by
contradiction: Let $t^*$ be a point such that $F(t^*) > G(t^*)$.  Clearly
$t^*\in (t_{j-1},t_j)$ for some $j$. Since $F(t_{j-1})\leq G(t_{j-1})$ by
assumption, it is necessary that $f(t)>g(t)$ for $t\in (t_{j-1},t_j)$. However
integrating from $t^*$ to $t_j$, we have that $F(t_j) > G(t_{j})$, which
contradicts the assumption that the inequality is valid for all $t_k$.
\end{proof}

The following lemma is well-known.

\begin{lemma}\label{lem:Mrs-Gerber} (Lemma 2 in \cite{wyz73})\\
The function
$h(x*h^{-1}(y))$ is strictly convex in $y$.  (  Key ingredient
of Mrs. Gerber's lemma)
\end{lemma}

\begin{lemma}(Lemma 1 in \cite{hap79})
 \\
\label{lem:Hardy-L-P}
 Let $x_1,...,x_l$ and $y_1,...,y_l$ be nondecreasing sequences of real numbers.
Let $\xi_1,...,\xi_l$ be a sequence of real numbers such that
$$ \sum_{j=k}^{l}\xi_j x_j \geq \sum_{j=k}^{l}\xi_j y_j, \quad 1\leq k\leq l $$
with equality for $k=1$. Then for any convex function $\Lambda$,
$$ \sum_{j=1}^{l} \xi_j \Lambda(x_j) \geq \sum_{j=1}^{l} \xi_j \Lambda(y_j). $$
\end{lemma}

\begin{theorem}\label{thm:suff-cond}
(A sufficient condition)\\
Given BISO channels $X\to Y$ and $X\to Z$ with Lorenz curves $F(t)$ and
$G(t)$, respectively. Further let $F(1)=G(1)$, i.e. channels have same capacity.
If $F(t)\leq G(t)$ then $Y$ is more
capable than $Z$.
\end{theorem}
\begin{proof}
Using Lemma \ref{lem:cont-major} we know that
$$ F(t_i) = \sum_{k=1}^{i} \xi_k f(t_k) \leq \sum_{k=1}^{i} \xi_k g(t_k) =
G(t_i),\quad i=1,\ldots, \hat N $$
and since $F(1) = G(1)$ we have equality at $i=\hat N$. Using Lemma
\ref{lem:Hardy-L-P} and by noticing that $f(t_k)$ and $g(t_k)$ are both
nondecreasing we have
$$ \sum_{j=1}^{\hat N} \xi_j \Lambda (f(t_j)) \geq \sum_{j=1}^{\hat N} \xi_j
\Lambda(g(t_j)) $$ for any convex function $\Lambda$. Taking $\Lambda(y) =
h(x*h^{-1}(y)) - y$ we obtain that
$$ \sum_{j=1}^{\hat N} \xi_j h\big(x*h^{-1}(f(t_j))\big) - \sum_{j=1}^{\hat N}
\xi_j f(t_j) \geq \sum_{j=1}^{\hat N} \xi_j h\big(x*h^{-1}(g(t_j))\big) -
\sum_{j=1}^{\hat N} \xi_j g(t_j).$$
From \eqref{eq:ixy} this is equivalent to
$$ I(X;Y)  \geq I(X;Z), \forall p(x). $$
Thus the theorem is established.
\end{proof}

For reasons that will be apparent later (Lemma \ref{th:mainmod}) it is useful to
zoom in on the following subclass of BISO channels.

Let $\mathcal{C}(C)$ be the class of BISO channels with capacity $C$.

For instance $BSC(p)$ belongs to this class, where $1-h(p) = C$. Similarly
$BEC(e)$ belongs to this class when $1-e=C$. Let $F(C)$ denote  an arbitrary BISO channel belonging to this class. 
Using an abuse of notation, we denote by $BSC(C)$ and $BEC(C)$ as the binary symmetric channel
and the binary erasure channel with capacity $C$, respectively.

\begin{corollary}
\label{cor:bsc}
$F(C) \gg BSC(C)$.
\end{corollary}
\begin{proof}
From Theorem \ref{thm:suff-cond} it suffices that the Lorenz curves
satisfy $G(t) \leq F_{BSC}(t), t \in [0,1]$.  Observe that $G(0)=F_{BSC}(0)=0$,
$G(1)=F_{BSC}(1)$ and that $F_{BSC}(t)$ is the straight-line connecting $0$ and
$F_{BSC}(1)$. The convexity of $G(t)$ (Property 1)  implies that $G(t) \leq
F_{BSC}(t), t \in [0,1]$.
\end{proof}

\begin{corollary}
\label{cor:bec}
$BEC(C) \gg F(C)$.
\end{corollary}
\begin{proof}
Similar to above it suffices that the Lorenz curves satisfy $F_{BEC}(t) \leq
G(t),  t\in [0,1]$. $F_{BEC}(t) = 0, t \in [0,1-e]$ and hence $F_{BEC}(t) \leq
G(t),  t\in [0,1-e]$. Combining $F_{BEC}(1) = G(1)$ and (comparing slopes)
$F_{BEC}'(t)=f_{BEC}(t)=1 \geq g(t) = G'(t), t \in (1-e,1]$, we also have
$F_{BEC}(t) \leq G(t),  t\in [1-e,1]$.
\end{proof}

\subsubsection{Relation to information combining}
Some of the results, more precisely Corollaries \ref{cor:bsc} and \ref{cor:bec},
can be obtained via an almost direct application of the results in \cite{ssz05}.
From \cite{ssz05}, for  $U\to X \sim BSC(s)$, if $Y$ is a BISO receiver (with
same capacity as BEC and BSC)
$$I(X;U,Y_{BSC}) \leq I(X;U,Y) \leq I(X;U,Y_{BEC})$$
which then yields $I(X;Y_{BSC}|U) \leq I(X;Y|U) \leq I(X;Y_{BEC}|U)$. But by
symmetry conditioning on $U$, where $U \to X \sim BSC(s)$ is same as taking $X
\sim \p(X=0)=1-s$. One could also obtain the same conclusion by using the
results in \cite{nai08b}.
However here we have used a different approach, via Theorem \ref{thm:suff-cond},
to establish the extreme properties of BSC and BEC.

\begin{corollary}\label{thm:card-3}
Let $F_1(C)$ and $F_2(C)$ be two BISO channels in $\mathcal{C}$ whose output
alphabet sizes are at most 3. Then either $F_1(C) \gg F_2(C)$ or $F_2(C) \gg
F_1(C)$, i.e. two such channels are always more capable comparable.
\end{corollary}
\begin{proof}
For BISO channel $X\to Y$ with transition probabilities $\{p_{-1},p_0,p_1\}$,
$k=0$ is split equally into $0^+$ and $0^-$. Thus the Lorenz curve $F(t)$
contains two sloping lines: one with slope
$h(\frac{q_{0^+}}{q_{0^+}+q_{0^-}})=1$, and the other not bigger than 1. Given
two Lorenz curves of this kind, $F(t)$ and $G(t)$, with $F(1)=G(1)$, then either
$F(t)\leq G(t)$ for all $t\in[0,1]$ or $F(t)\geq G(t)$ for all $t\in[0,1]$
(Figure~\ref{fig:card-3}). According to Theorem~\ref{thm:suff-cond}, these two
channels are more capable comparable.
\end{proof}

\begin{figure}[t]\centering
\includegraphics[width=0.6\textwidth]{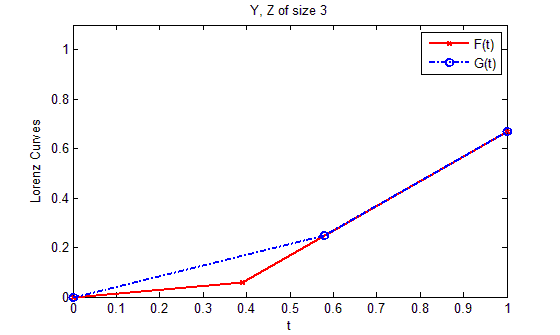}
\caption{Lorenz curves for BISO channels with the same capacity and
output of size 3.}
 \label{fig:card-3}
\end{figure}

\begin{remark}
\label{re:cap} Not all BISO channels with the same capacity are more capable
comparable. A counter example is the following: Consider a BISO channel $X \to
(Y,Z)$ with  transition probabilities according to:
\begin{align*}
\p(Y=i|X=0) &= a_i, -2\leq i\leq 2\\
\p(Z=j|X=0) &= b_j, -2\leq j\leq 2
\end{align*}
where $a_{-2} = 0.061, a_{-1} = a_{1}=\frac{1-10a_{-2}}{2},a_{2} = 9a_{-2}$ and
$b_{-2} = 0.0634977, b_{-1} = \frac{1-b_{-2}}{5}, b_{1} =
\frac{4(1-b_{-2})}{5},b_{2} = 0.$ One can verify that the channels have same
capacity, but are not more capable comparable.
\end{remark}

\medskip

\subsubsection{On more capable and essentially less noisy orderings in BISO channels}

In this section we will establish that these two partial orderings, restricted
to $\mathcal{C}$, are inverses of each other(!). This is counter-intuitive as
more capable and essentially less noisy are two notions of saying that one
receiver is superior to another receiver.

Below (for a complete argument see Lemma 1 in \cite{nai08b}) we note that the
uniform input distribution forms a sufficient class for a broadcast channel
consisting of two channels $F_1, F_2 \in \cC$.
\begin{claim}
\label{cl:suff}
Consider a binary input broadcast channel whose component channels, $F_1: X \to
Y_1$ and $F_2: X \to Y_2$ are both output-symmetric, i.e. $F_1,F_2 \in \cC$.
Then the uniform input distribution $\p(X=0) = \frac 12$ forms a sufficient
class.
\end{claim}
\begin{proof}
The following construction suffices - we leave the details to the reader. Let
$j,k \in \{0,1\}$; then define $$\textrm{Q}(U=(u,j),V=(v,k),X=x) = \begin{cases}
\begin{array}{ll} \frac 12  \textrm{P}(U=u,V=v,X=x \oplus j ) & j=k \\ 0 & j
\neq k \end{array} \end{cases}.$$
\end{proof}

\smallskip

\begin{lemma}
\label{le:eq}
Let $F_1,F_2 \in \cC(C)$; then $F_1 \gg F_2 \iff F_2 \succeq F_1$.
\end{lemma}
\begin{proof}
Assume $F_1 \gg F_2$.
From Claim \ref{cl:suff} we know that $\p(X=0) = \frac 12$ is a
sufficient distribution for the channels $F_1,F_2$. Therefore, when $\p(X=0) =
\frac 12$ we have for all $U$ such that $U \to X \to (Y_1,Y_2)$
\begin{align*}
I(U;Y_1) &= I(X;Y_1) - I(X;Y_1|U) \\
& = C - I(X;Y_1|U) \\
& = I(X;Y_2) - I(X;Y_1|U) \\
& = I(U;Y_2) + I(X;Y_2|U) - I(X;Y_1|U) \\
& \leq I(U;Y_2),
\end{align*}
where the last inequality follows  from $F_1 \gg F_2$. Since
$\p(X=0) = \frac 12$ is a sufficient class of input distributions for a
broadcast channel comprising of $F_1,F_2$ it follows from the definition that
$F_2 \succeq F_1$.

\smallskip

Assume $F_2 \succeq F_1$. The proof follows by contradiction. Suppose there is a
value $x$ such that when $\p(X=0) = x, I(X;Y_2) - I(X;Y_1)=\delta > 0$, then
consider a $U$ such that $\p(U=0)=\p(U=1) = \frac 12$, $\p(X=0|U=0) = x =
\p(X=1|U=1)$. Observe that, from the symmetry $I(X;Y_2|U) - I(X;Y_1|U) = \delta
> 0$. However since $\p(X=0) = \frac 12$, using a similar decomposition we see
that
\begin{align*}
I(U;Y_1) & = I(U;Y_2) + I(X;Y_2|U) - I(X;Y_1|U) \\
& = I(U;Y_2) + \delta > I(U;Y_2),
\end{align*}
contradicting the assumption $F_2 \succeq F_1$. Therefore $F_1 \gg F_2$.
\end{proof}

The following lemma is an immediate consequence of Corollaries \ref{cor:bsc},
\ref{cor:bec}, and Lemma \ref{le:eq}.
\begin{lemma}
\label{th:mainmod}
Let $BSC(C)$ represent a binary symmetric channel with capacity $C$, $BEC(C)$ -
a binary erasure channel with capacity $C$, and $F(C)$ - an arbitrary binary
input symmetric output channel, i.e. $F \in \cC$, with capacity $C$. We have
\begin{itemize}
\item[$(i)$] $BEC(C) \gg F(C) \gg BSC(C)$,
\item[$(ii)$] $BSC(C) \succeq F(C) \succeq BEC(C)$.
\end{itemize}
\end{lemma}

This leads us to one of the main results in this paper.

\begin{theorem}
\label{th:main}
Let $BSC(C)$ represent a binary symmetric channel with capacity $C$, $BEC(C)$ -
a binary erasure channel with capacity $C$, and $F(C)$ - an arbitrary binary
input symmetric output channel, i.e. $F \in \cC$, with capacity $C$. For any
three numbers $0 \leq C_1 \leq C_2 \leq C_3$ we have
\begin{itemize}
\item[$(i)$] $BEC(C_3) \gg F(C_2) \gg BSC(C_1)$,
\item[$(ii)$] $BSC(C_3) \succeq F(C_2) \succeq BEC(C_1)$.
\end{itemize}
\end{theorem}

\begin{proof}
If $C_a < C_b$ then $BSC(C_a) ,BEC(C_a)$ are degraded versions of $BSC(C_b),
BEC(C_b)$ respectively. Hence from Lemma \ref{th:mainmod} we
have
$$ BEC(C_3) \gg BEC(C_2) \gg  F(C_2) \gg BSC(C_2) \gg BSC(C_1), $$
$$ BSC(C_3) \succeq BSC(C_2) \succeq  F(C_2) \succeq BEC(C_2) \succeq BEC(C_1). $$
\end{proof}

The following corollary is immediate.

\begin{corollary}
\label{cor:sup}
Superposition coding region is the capacity region for a BISO-broadcast channel
if any one of the channels is either a BSC or a BEC.
\end{corollary}
\begin{proof}
Superposition coding is optimal  both for more capable
comparable channels\cite{elg79} and for essentially less noisy comparable channels \cite{nai08b}. From Theorem \ref{th:main}, if any one of
the channels is either a BSC or a BEC, then the channels are either more capable
comparable or essentially less noisy comparable.
\end{proof}

\begin{remark}
In \cite{nai08b} the capacity region of a BSC/BEC broadcast channel was
established.  Corollary  \ref{cor:sup} generalizes this result to only requiring
that one of the BISO channels is a BEC or a BSC.
\end{remark}

\subsection{Comparison of inner and outer bounds for BISO channels}

The following are some commonly used inner bounds (or achievable rate regions)
for the capacity region (CR):
\begin{itemize}
\item Time-Division region (TD): This region is characterized by the set of points
\begin{align*}
R_1 &\leq \alpha C_1 \\
R_2 &\leq (1-\alpha) C_2,
\end{align*}
where $C_1$ and $C_2$ are the channel capacities for the two receivers,
respectively. The rates are achieved by transmitting at capacity $C_1$ to the
first receiver for  fraction $\alpha$ of the time, and at capacity $C_2$ to
second receiver for the remaining fraction.
\item Randomized Time-Divison region (RTD): This corresponds to a time-division
strategy except that the slots for which communication occurs to one receiver 
is also drawn from a codebook which conveys additional information. The rates
are characterized by
\begin{align*}
R_1 &\leq I(W;Y_1) + \p(W=0)I(X;Y_1|W=0) \\
R_2 &\leq I(W;Y_2) + \p(W=1)I(X;Y_2|W=1) \\
R_1 + R_2 & \leq \min \{I(W;Y_1), I(W;Y_2) \} + \p(W=0)I(X;Y_1|W=0) + \p(W=1)I(X;Y_2|W=1),
\end{align*}
over binary random variables $W$ satisfying $W \to X \to (Y_1,Y_2)$ being Markov.
The binary random variable $W$ characterizes the slots which distinguish
communication to one receiver over the other.
\item Marton's Inner bound (MIB): This is the best known achievable rate region.
The rates are characterized by
\begin{align*}
R_1 &\leq I(U,W;Y_1)  \\
R_2 &\leq I(V,W;Y_2)  \\
R_1 + R_2 & \leq \min \{I(W;Y_1), I(W;Y_2) \} + I(U;Y_1|W) + I(V;Y_2|W) - I(U;V|W),
\end{align*}
over random variables $(U,V,W)$ satisfying $(U,V,W) \to X \to (Y_1,Y_2)$ being Markov.
Observe that setting $U=X,V=\emptyset$ when $W=0$ and $V=X,U=\emptyset$ when
$W=1$ reduces MIB to the RTD region.
\begin{lemma}[\cite{nwg10}] For binary input broadcast channels, the
maximum sum rate implied by Marton's inner bound(MIB) matches that of  randomized time-divison(RTD) region.
\label{le:rtdmib}
\end{lemma}
\item Outer bound (OB): The following region\cite{nae07} represents an outer
bound to the capacity region.  The union of rate pairs
\begin{align*}
R_1 &\leq I(U;Y_1) \\
R_2 &\leq I(V;Y_2) \\
R_1 + R_2 & \leq I(U;Y_1) + I(X;Y_2|U) \\
R_1 + R_2 & \leq I(V;Y_2) + I(X;Y_1|V)
\end{align*}
over all $(U,V) \to X \to (Y_1,Y_2)$ represents an outer bound to the capacity region.
\begin{remark}
\label{rem:km}
For BISO channels since $\p(X=0) = \frac 12$ is a common sufficient
distribution, it can be shown that the OB matches an earlier outer bound due to
 K{\"{o}}rner and Marton \cite{mar79}.
\end{remark}
\end{itemize}

We adopt the notation in Table \ref{table:notation}.
\begin{table}[t]\caption{Notation}\label{table:notation}
\smallskip \centering
\begin{tabular}{|c|c| c|c|}\hline
Abbr. & & Abbr. & \\\hline
 TD  & time-division region & BSC & binary symmetric channel\\
 RTD & randomized time-division region & BEC & binary erasure channel\\
 MIB  & Marton's inner bound & e.l.n. & essentially less noisy\\
 CR  & capacity region & e.m.c. & essentially more capable\\
 OB  & Outer bound (K{\"{o}}rner-Marton, Nair-El
Gamal) & $*$ & binary convolution \\
 BISO & binary input symmetric output & $h(\cdot)$ & binary entropy function\\
\hline
\end{tabular}
\end{table}

%%%%%%%%%%%%%%%%%%%%%%%%%%%%%%%%%%%%%%%%%%%%%%%%%%%%%%%%%%%%%%%%%%%%%%%%%%

\medskip
\begin{lemma}\label{lem:BSC-aux}
Consider a 2-receiver broadcast channel where both $X \to Y_1$
and $X \to Y_2$ represent the BISO channels with transition
probabilities $\{q_k,q_{-k}: 1 \leq k \leq N\}$ and $\{p_j,p_{-j}: 1 \leq j \leq
N\}$ respectively. Consider the following region formed by taking the union
of rate pairs $(R_1,R_2)$ satisfying
\begin{align*}
R_2 &\leq I(U;Y_2)\\
R_2+R_1 &\leq I(U;Y_2)+ I(X;Y_1|U)\\
R_1 &\leq I(X;Y_1)
\end{align*}
over all $p(u)p(x|u)p(y_1,y_2|x)$. Then the same region can be realized by
restricting to a binary $U$ such that $U\to X \sim BSC(s)$ and $\p(X=0)=\frac
12$.
\end{lemma}

\begin{proof}
The proof is presented in the Appendix.
\end{proof}

Let $U\to X \sim BSC(s_1), V\to X \sim BSC(s_2)$ and $\p(X=0)=\frac 12$. Let $I(U;Y_1)
= f_1(s_1)$, where $\p(X=1|U=0)=s_1$, and define $I(V;Y_2)=f_2(s_2)$ in a
similar fashion. It is clear from symmetry that $f_1(s) = f_1(1-s), f_2(s)=f_2(1-s).$

From Lemma \ref{lem:BSC-aux} and Remark \ref{rem:km} it follows that OB
can be written as the union of rate pairs $R_1,R_2$ satisfying
\begin{align}
R_1 &\leq f_1(s_1) \nonumber \\
R_2 & \leq f_2(s_2) \nonumber \\
R_1 + R_2 &\leq f_1(s_1) + C - f_2(s_1) \label{eq:OBeq}\\
R_1 + R_2 &\leq f_2(s_2) + C - f_1(s_2). \nonumber
\end{align}
for some $0 \leq s_1,s_2 \leq \frac 12$.

Let
\begin{align*}
 I&=\{ s\in[0,0.5]: f_1(s)>f_2(s)\}\\
 J&=\{ s\in[0,0.5]: f_1(s)<f_2(s)\}.
\end{align*}

The following result relates the equivalence of the various bounds and their
relation to whether the channels are more capable comparable.

% \begin{lemma}
% \label{lem:mceq}
% The two channels are not more capable comparable {\em if and only if}
% there exists $s_1\in I, s_2\in J $ such that $f_1(s_1)+f_2(s_2)>C$.
% \end{lemma}
% \begin{proof}
% When $s_1 \in I$ we have the following:
% Let $\p(X=0)=\frac 12$. Then
% $$I(X;Y_1|U) = I(X;Y_1) - I(U;Y_1)  = C-f_1(s_1) < C-f_2(s_1) = I(X;Y_2) - I(U;Y_2) = I(X;Y_2|U). $$
% Similarly we can obtain that when $s_2 \in J$, we have $I(X;Y_2|V) < I(X;Y_1|V)$. Thus if there exist two points $s_1\in I, s_2 \in J$ then the channels are not more capable comparable. The proof of the reverse direction is presented in the appendix.
% \end{proof}

\begin{theorem}\label{thm:TD-RTD-OB}
Let $F_1, F_2 \in \mathcal{C}(C)$. Then the following are equivalent:
\begin{enumerate}[(a)]
\item $F_1$ and $F_2$ are not more capable comparable
\item $TD \subset OB$ \label{it:comp4}
\item There exists $s_1\in I, s_2\in J $ such that $f_1(s_1)+f_2(s_2)>C$
\item $TD \subset MIB$ \label{it:comp2}
\item $MIB \subset OB$. \label{it:comp3}
\end{enumerate}
\end{theorem}

\begin{proof}
The proof of this equivalence is presented in the Appendix.
\end{proof}

\medskip
%%%%%%%%%%%%%%%%%%%%%%%%%%%%%%%%%%%%%%%%%%%%%%%%%%%%%%%%%%%%%%%%%%%%%%%%%%

\begin{corollary}\label{thm:BISO-n-comp}
For two BISO channels with the same capacity, superposition coding is optimal if
and only if the channels are more capable comparable.
\end{corollary}
\begin{proof}
If superposition coding region is indeed the capacity region, then we have $R_1
+ R_2 \leq I(X;Y_1) \leq C$. Further since the two channels have the same
capacity, we have the TD region is optimal. From Theorem \ref{thm:TD-RTD-OB} we
have that the channels are more capable comparable.
\end{proof}

\begin{remark}
A characterization of when superposition coding is optimal for 2-receiver
broadcast channels is open in general. It is known that superposition coding is
optimal  when the channels are either essentially more capable comparable or
essentially less noisy comparable\cite{nai08b} - two incompatible notions.
However a converse statement is still unknown.
\end{remark}

\begin{observation}
\label{obs:mc}
From remark \ref{re:cap} we know that there exists a pair of channels
$F_1, F_2 \in \mathcal{C}(C)$ which are not more capable comparable. Hence from
Theorem \ref{thm:TD-RTD-OB} we know that the capacity region is strictly larger
than TD. However, if we replace $F_2$ by $BEC(C)$, a more capable channel, then
the capacity of the broadcast channel formed by $F_1$ and $BEC(C)$ is the TD
region (Corollary \ref{cor:bec}). Thus replacing by a more capable channel can
{\em strictly} reduce the capacity region.
\end{observation}

This observation leads to an operational definition of a better receiver and a
partial order as follows.

\subsubsection{A new partial order}

We now introduce a  natural operational partial order among broadcast channels.

\begin{definition}
Receiver $Z_2$ is a {\em better receiver} than $Y_2$ if the capacity region of
$X\to (Y_1, Z_2)$ contains that of $X\to (Y_1, Y_2)$ for every channel $X \to
Y_1$. In other words, if we replace receiver $Y_2$ by receiver $Z_2$ then the
capacity region will not decrease.
\end{definition}

\begin{remark}
Note that the capacity region of a broadcast channel just depends on the
marginal distributions $X \to Y_1$, $X \to Y_2$, and hence the definition makes
sense.
\end{remark}

From Observation \ref{obs:mc} we know that a more capable receiver is
not necessarily a better receiver. However we will show that if $Z_2$ is a less
noisy receiver than $Y_2$, then $Z_2$ is indeed a better receiver than $Y_2$.

\begin{claim}
\label{cl:ln}
If $Z_2$ is a less noisy receiver than $Y_2$, then $Z_2$ is a better receiver
than $Y_2$.
\end{claim}

\begin{proof}
The capacity region of a discrete memoryless broadcast channel has the following
$n$-letter characterization. Consider the region $\mathcal{R}_n$ defined as the
union of rate pairs $(R_1,R_2)$ that  satisfy
\begin{align*}
R_1 \leq \frac 1n I(U;Y_1^n) \\
R_2 \leq \frac 1n I(V;Y_2^n)
\end{align*}
for some $p(u)p(v)p(x^n|u,v)$. It is known that the capacity region is $\lim_n
\mathcal{R}_n$. (This is folklore. It is clear that this is achievable, and a
converse follows by setting $U=M_1$ and $V=M_2$ and applying Fano's inequality.)
Observe that
\begin{align*}
I(V;Y_{2,1}^{~j},Z_{2,j+1}^{~n}) & =I(V;Y_{2,1}^{j-1},Z_{2,j+1}^{~n}) + I(V;Y_{2j}|Y_{2,1}^{j-1},Z_{2,j+1}^{~n}), ~j=n,\ldots,1 \\
& \leq I(V;Y_{2,1}^{j-1},Z_{2,j+1}^{~n}) + I(V;Z_{2j}|Y_{2,1}^{j-1},Z_{2,j+1}^{~n}) \\
& = I(V;Y_{2,1}^{j-1},Z_{2,j}^{~n}).
\end{align*}
By taking the extreme points of this chain we obtain that $I(V;Y_2^n) \leq I(V;Z_2^n)$.
Claim follows from the expression of the capacity region stated above.
\end{proof}

\section{Conclusion}
We look at partial orders induced by the more capable relations and less noisy
relations in binary-input symmetric-output(BISO) broadcast channels. We
establish the capacity regions for a class of them and also show various other
results related to the evaluation of various bounds. Some of the results act
contrary to popular intuition and hence BISO channels can serve as a simple
class from which we can improve our understanding of various relations. We also
use perturbation based arguments to show the optimality of certain auxiliary
channels, thus generalizing earlier results. We hope that some of the results
presented here can invoke a careful rethinking of various notions of dominance
between receivers.

\bibliographystyle{amsplain}
\bibliography{mybiblio}

\appendix

\subsection{Proof to Lemma~\ref{lem:BSC-aux}}
\begin{proof}
Let $\cU=\{1,2,...,m\}$, $\p(U=i) = u_i$ and $\p(X=0|U=i)=s_i$.
Further let $h(x) = -x \log_2x - (1-x) \log_2(1-x) $ be the binary
entropy function and let $*$ denote the binary convolution, i.e. $a*b
= a(1-b)+b(1-a)$.

Using these notations we have the following expansions,
\begin{align*}
 I(U;Y_2) &= \sum_j (p_j + p_{-j}) \big(h(\frac {p_j}{p_j + p_{-j}}* \sum_i u_i
s_i) - \sum_i u_i h(\frac {p_j}{p_j + p_{-j}}*s_i) \big) \\
 I(X;Y_1|U) &= \sum_k (q_k + q_{-k}) \big(\sum_i u_i h(\frac {q_k}{q_k +
q_{-k}}* s_i) - h(\frac {q_k}{q_k + q_{-k}}) \big)  \\
 I(X;Y_1) &= \sum_k (q_k + q_{-k}) \big( h(\frac {q_k}{q_k + q_{-k}}* \sum_i u_i
s_i) - h(\frac {q_k}{q_k + q_{-k}}) \big).
\end{align*}

Define $\tilde{\cU} = \{1,2,...,m\} \times \{1,2\}$,
$\p(\tilde{U}=(i,1)) = \frac{u_i}{2}$,
$\p(X=0|\tilde{U}=(i,1))=s_i$, $\p(\tilde{U}=(i,2)) =
\frac{u_i}{2}$, and $\p(X=0|\tilde{U}=(i,2))=1-s_i$. This induces an
$\tilde{X}$ with $\p(\tilde{X}=0)=\frac{1}{2}$ and it is
straightforward to notice
\begin{align*}
 I(\tilde U; \tilde Y_2) &\geq I(U;Y_2), \\
 I(\tilde X;\tilde Y_1|\tilde U) &= I(X;Y_1|U) ,  \\
 I(\tilde X;\tilde Y_1) &\geq I(X;Y_1).
\end{align*}

Thus for every $U$ replacing $U$ by $\tilde{U}$ leads to a larger
achievable region.

Hence it suffices to maximize over all auxiliary
random variables of the form $(U,X)$ defined by: $\cU = \{1,2,...,m\}
\times \{1,2\}$, $\p(U=(i,1)) = \frac{u_i}{2}$,
$\p(X=0|U=(i,1))=s_i$, $\p(U=(i,2)) = \frac{u_i}{2}$ and
$\p(X=0|U=(i,2))=1-s_i$. Let this class of random variables $(U,X)$ be $\mathcal{Q}$.

Since $\p(X=0) = \frac 12$ remains fixed, the third inequality remains constant.
Therefore, to compute the extreme points, we proceed to compute the distribution
$(U,X)$ (belonging to $\mathcal{Q}$) that maximizes $\lambda I(U;Y_2) +
\big(I(U;Y_2) + I(X;Y_1|U)\big).$

For a given $p(u,x) \in \mathcal{Q}, |\mathcal{U}|=2m$ , consider the
multiplicative Lyapunov perturbation defined by
\begin{align}
 R(U=(i,1),X=0)&= \p(U=(i,1),X=0)(1+ \varepsilon L(i)) \nonumber\\
 R(U=(i,1),X=1)&= \p(U=(i,1),X=1)(1+ \varepsilon L(i)) \label{eq:cond}\\
 R(U=(i,2),X=0)&= R(U=(i,1),X=1) \nonumber\\
 R(U=(i,2),X=1)&= R(U=(i,1),X=0) \nonumber
\end{align}
For $r(u,x)$ to be a valid probability distribution we require the
conditions $1+\varepsilon L(i) \geq 0, \forall i $ and $\sum_{=1}^mi \p
\big(U=(i,1)\big)L(i) = 0.$

Observe that the perturbation maintains $\p(X=0)$ and further the new pair
$r(u,x)$ also belongs to $\mathcal{Q}$. A  non-trivial $L$ exists if
$m=\frac{|\mathcal U|}{2} \geq 2 $.

Observe that
\begin{align*}
& (\lambda+1) I_{r}(U;Y_2)+ I_{r}(X;Y_1|U)\\
&\quad  = (\lambda+1) H_p(Y_2) + \lambda H_p(U) + H_{p}(U,Y_1) -
(\lambda+1) H_{p}(U,Y_2)\\
&\qquad + \varepsilon \big(\lambda H^L_p(U)+ H^L_{p}(U,Y_1) -
(\lambda+1) H^L_{p}(U,Y_2)\big)
\end{align*}
where
\begin{align*}
H^L_p(U) &= -\sum_{i} 2p(i)L(i) \log 2p(i) \\
H^L_p(U,Y_1) &= -\sum_{i,y_1} 2p(i,y_1)L(i) \log 2p(i,y_1) \\
H^L_p(U,Y_2) &= -\sum_{i,y_2} 2p(i,y_2)L(i) \log 2p(i,y_2).
\end{align*}
The first derivative with respect to $\varepsilon$ being zero
implies $$\lambda H^L_p(U)+ H^L_{p}(U,Y_1) - (\lambda+1)
H^L_{p}(U,Y_2) = 0$$ and this further implies that if $p(u,x)$
achieves the maximum of $ (\lambda+1) I_{p}(U;Y_2)+ I_p(X;Y_1|U)$
then $ (\lambda+1) I_{r}(U;Y_2)+ I_{r}(X;Y_1|U) = (\lambda+1)
I_{p}(U;Y_2)+ I_p(X;Y_1|U)$ for any valid perturbation that
satisfies \eqref{eq:cond}.

Now we choose $\varepsilon$ such that $\min_i 1+\varepsilon L(i)
=0$, and let $i=i^*$ achieve this minimum. Observe that $r(i^*) = 0$
and hence there exists an $U$ with cardinality equal to $2 (m-1)$ such that
$(\lambda+1) I(U;Y_2)+ I(X;Y_1|U)$ is constant. We can proceed by induction
until $m = 1$.

Since $(U,X) \in \mathcal{Q}$ and $|\mathcal{U}|=2$, implies that the optimal
auxiliary channel $U \to X$ follows the distribution given by
\begin{align*}
\p(U=1) &= \p(U=2)=\frac 12 \\
\p(X=0|U=1) &= \p(X=1|U=2) = s,
\end{align*}
i.e. $U \to X \sim BSC(s)$.
\end{proof}

The same proof can also be used to establish the following lemma.

\begin{lemma}\label{lem:BSC-aux-deg}
Consider a 2-receiver broadcast channels where both $X \to Y_1$
and $X \to Y_2$ represent the BISO channels with transition
probabilities $\{q_k,q_{-k}: 1 \leq k \leq N\}$ and $\{p_j,p_{-j}: 1 \leq j \leq
N\}$ respectively. Consider the following superposition coding region formed by
taking the union of rate pairs $(R_1,R_2)$ satisfying
\begin{align*}
R_2 &\leq I(U;Y_2)\\
R_2+R_1 &\leq I(U;Y_2)+ I(X;Y_1|U)\\
R_2+R_1 &\leq I(X;Y_1)
\end{align*}
over all $p(u)p(x|u)p(y_1,y_2|x)$. Then the same region can be realized by
restricting to a binary $U$ such that $U\to X \sim BSC(s)$ and $\p(X=0)=\frac
12$.
\end{lemma}

\begin{remark}
This generalizes the result by Wyner and Ziv \cite{wyz73} for BSC broadcast
channels. In \cite{elg79} it was shown that superposition coding is indeed
optimal when the two channels are more capable comparable.
\end{remark}

\subsection{Proof to Theorem~\ref{thm:TD-RTD-OB}}

\begin{proof}

(a) $\Rightarrow$ (b):
Recalling: Let
\begin{align*}
 I&=\{ s\in[0,0.5]: f_1(s)>f_2(s)\}\\
 J&=\{ s\in[0,0.5]: f_1(s)<f_2(s)\}.
\end{align*}
Since the channels are not more-capable comparable, we know that there esists
$s_1 \in I$ and $s_2 \in J$. Construct $\tilde{U}\to X$, where
$\tilde{U}=U'\times Q$ with binary $U'$ and $Q$, and probabilities
\begin{flalign*}
 & & \p(\tilde{U}=(0,0)) &= \frac{1-\varepsilon}{2}  & \p(X=0|\tilde{U}=(0,0))
&= 1   & \\
 & & \p(\tilde{U}=(0,1)) &= \frac{\varepsilon}{2}    & \p(X=0|\tilde{U}=(0,1))
&= s_1 & \\
 & & \p(\tilde{U}=(1,0)) &= \frac{1-\varepsilon}{2}  & \p(X=1|\tilde{U}=(1,0))
&= 1   & \\
 & & \p(\tilde{U}=(1,1)) &= \frac{\varepsilon}{2}    & \p(X=1|\tilde{U}=(1,1))
&= s_1. &
\end{flalign*}
Thus, $U'\mapsto X \sim BSC(0)$ conditioned on the event $Q=0$,  $ U'\mapsto X
\sim BSC(1-s_1)$ conditioned on $Q=1$, and further $U'$ is independent of $Q$
with $\p(U'=0) = \frac 12$. We can see that $Q$ is independent of $X$ and hence
of $Y_1,Y_2$; thus $I(Q;Y_1) = I(Q;Y_2) =
0$. Now
\begin{align*}
 I(\tilde{U};Y_1) &= I(U',Q;Y_1) = I(U';Y_1|Q)+I(Q;Y_1)\\
&=I(U';Y_1|Q)\\
&=(1-\varepsilon)I(X;Y_1) +\varepsilon I(U';Y_1|Q=1)\\
&=(1-\varepsilon)C + \varepsilon f_1(s_1).
\end{align*}

Similarly, we obtain
$$ I(\tilde{U};Y_2) = (1-\varepsilon)C + \varepsilon f_2(s_1).$$

Thus we have
\begin{align*}
 R_1&\leq (1-\varepsilon) C+\varepsilon f_1(s_1)\\
R_2&\leq f_2(s_2) \\
R_1+R_2 &\leq I(\tilde{U};Y_1) +I(X;Y_2|\tilde{U})\\ &=
I(\tilde{U};Y_1) +I(X;Y_2)-I(\tilde{U};Y_2) \\
&=(1-\varepsilon)C+\varepsilon f_1(s_1)+C-[(1-\varepsilon)C+\varepsilon
f_2(s_1)]\\ &=C+\varepsilon[f_1(s_1)-f_2(s_1)]
\quad\quad (> C) \\
R_1+R_2 &\leq I(V;Y_2)+I(X;Y_1|V)\\ &= f_2(s_2)+C-f_1(s_2)
\quad\qquad (> C).
\end{align*}
To show that we can have $(1-\varepsilon) C+\varepsilon f_1(s_1)
+f_2(s_2)>C$, we just need to choose small $\varepsilon$ to ensure
$f_2(s_2)>\varepsilon [C-f_1(s_1)]$. Since this is clearly possibe, we have $OB \supset TD$.

\medskip

(b) $\Rightarrow$ (c):
From Equation \eqref{eq:OBeq}, we
have the following expression of the boundary of the outer bound,
\begin{align*}
 R_1 &\leq I(U;Y_1) = f_1(s_1)\\
 R_2 &\leq I(V;Y_2) = f_2(s_2)\\
 R_1+R_2 &\leq I(U;Y_1)+I(X;Y_2|U)= f_1(s_1) + C - f_2(s_1)\\
 R_1+R_2 &\leq I(V;Y_2)+I(X;Y_1|V)= f_2(s_2) + C - f_1(s_2)
\end{align*}

Clearly for every $s_1 \in I, s_2 \in J$ if $f_1(s_1) + f_2(s_2) \leq C$ then
from above $OB=TD$. However since $OB \supset TD$, there exists $s_1\in I,
s_2\in J $ such that $f_1(s_1)+f_2(s_2)>C$.

\medskip

(c) $\Rightarrow$ (d):
In general, $TD \subseteq RTD \subseteq MIB$. So now it suffices to show there
exists an example where the sum rate of RTD region is strictly
larger than TD region. 

We now compute the maximum sum rate of the RTD region. From Lemma \ref{le:rtdmib}
we know that this matches the 
maximum sum rate of the MIB region.

Consider an
auxiliary channel $W \to X$ such that
\begin{align*}
\p(W=0) = a, \quad &\p(W=1)=1-a \\
\p(X=0|W=0) =s_2, \quad &\p(X=0|W=1) = s_1
\end{align*}
where $as_2+(1-a)s_1 = \frac 12$.

It is straightforward to check the following
\begin{align*}
& I(X;Y_1|W=0) = C-f_1(s_2), ~ I(X;Y_1|W=1) = C-f_1(s_1) \\
& I(X;Y_2|W=0) = C- f_2(s_2), ~ I(X;Y_2|W=1) = C-f_2(s_1), \\
& I(X;Y_1) = I(X;Y_2) = C.
\end{align*}

Then observe that
\begin{align*}
& I(W;Y_1) + \p(W=0)I(X;Y_1|W=0)+ \p(W=1)I(X;Y_2|W=1)\\
& \quad = I(X;Y_1) + \p(W=1)\big(I(X;Y_2|W=1) - I(X;Y_1|W=1) \big) \\
& \quad = C + (1-a) (f_1(s_1) - f_2(s_1)) 
\end{align*}
where the last inequality holds since $s_1 \in I$. 

Similarly
$$ I(W;Y_2) + \p(W=0)I(X;Y_1|W=0)+ \p(W=1)I(X;Y_2|W=1) = C + a (f_2(s_2) - f_1(s_2)).$$

Therefore the sum rate of RTD (eq. MIB) for this choice of $(W,X)$ is given by
\begin{equation}
 C + \min \{ (1-a) (f_1(s_1) - f_2(s_1)), a (f_2(s_2) - f_1(s_2)) \}. 
\label{eq:maxrtd}
\end{equation}

Therefore if $(c)$ is satisfied, i.e. there exists $s_1 \in I, s_2 \in J$, then
there exists a $(W,X)$ so that equation \eqref{eq:maxrtd} gives a sum rate
strictly larger than $C$. 

\begin{remark}
A careful reader will notice that the above argument only requires  $s_1 \in I,
s_2 \in J$ and does not even require $f_1(s_1) + f_2(s_2) > C$. But existence of
any $s_a \in I, s_b \in J$ will imply that $(a)$ holds and hence $(c)$ holds.
\end{remark}

\medskip

% (4) $\Rightarrow$ (5):
% 
% 
% \medskip

(d) $\Rightarrow$ (e):
Since $TD \subset MIB$, to compute the maximum sum rate of MIB it suffices to
maximize over $s_1 \in I, s_2 \in J, 0 < a < 1$ the term
$$  C + \min \{ (1-a) (f_1(s_1) - f_2(s_1)), a (f_2(s_2) - f_1(s_2)) \}. $$

Consider any triple $s_1 \in I, s_2 \in J, 0 < a < 1$. Pick any $\varepsilon >
0$ small enough (will show later how small we require it).

Define $(U,X)=(Q,U_1,X)$ where $\p(Q=0) = 1- a + \varepsilon, \p(Q=1) = a -
\varepsilon$; and $U_1 \mapsto X \sim BSC(s_1)$ conditioned on $Q=0$,
and $U_1 \mapsto X \sim BSC(0)$ conditioned on $Q=1$. Further take $\p(U_1=0) =
\p(U_1=1) = \frac 12$. Observe that this induces $\p(X=0) = \p(X=1) = \frac 12$.

Similarly define $(V,X)=(Q',V_1,X)$ where $\p(Q'=0) = a + \varepsilon, \p(Q'=1)
= 1 - a - \varepsilon$; and $V_1 \mapsto X \sim BSC(s_2)$ conditioned on $Q'=0$,
and $V_1 \mapsto X \sim BSC(0)$ conditioned on $Q'=1$. Further take $\p(V_1=0) =
\p(V_1=1) = \frac 12$. Observe that this also induces $\p(X=0) = \p(X=1) = \frac
12$.

Since the distribution of $X$ is consistent there exists a triple $(U,V,X)$ with
the same pairwise marginals $(U,X)$ and $(V,X)$ as described earlier.
With this choice, OB reduces to 
\begin{align*}
R_1 &\leq I(U;Y_1) = (1 - a+\varepsilon) f_1(s_1) + (a-\varepsilon) C \\
R_2 & \leq I(V;Y_2) = (a +\varepsilon) f_2(s_2) + (1 - a - \varepsilon) C \\
R_1 + R_2 & \leq  I(U;Y_1) + I(X;Y_2|U) = C + (1 - a + \varepsilon) (f_1(s_1) - f_2(s_1)) \\
R_1 + R_2 & \leq  I(V;Y_2) + I(X;Y_1|V) = C + (a + \varepsilon) (f_2(s_2) - f_1(s_2)).
\end{align*}

Clearly the maximum sum rate of the above region is minimum of the terms 
$$\{ C + (1 - a + \varepsilon) (f_1(s_1) - f_2(s_1)), C + (a + \varepsilon)
(f_2(s_2) - f_1(s_2)), (1-2\epsilon) C + (1 - a+\varepsilon) f_1(s_1) + (a
+\varepsilon) f_2(s_2) \}.$$

We pick $\varepsilon > 0$ to satisfy 
\begin{align*}
 &   (1-2\epsilon) C + (1 - a+\varepsilon) f_1(s_1) + (a +\varepsilon) f_2(s_2)
> C + (1-a) (f_1(s_1) - f_2(s_1))\\
\Leftrightarrow \quad  &  (1-a) f_2(s_1) + a f_2(s_2) > \varepsilon( 2C - f_1(s_1) - f_2(s_2) ),
\end{align*}
and 
$$ a f_1(s_2) + (1-a) f_2(s_1) > \varepsilon( 2C - f_1(s_1) - f_2(s_2) ), $$
then the maximum sum rate of the OB expression will be strictly bigger than that
of MIB region. Since this is possible for every $s_1 \in I, s_2 \in J, 0 < a <
1$, the maximum sum rate of OB is strictly larger than that of MIB. Therefore $OB
\supset MIB$ or $(e)$ holds.

\medskip

(e) $\Rightarrow$ (a):
Since $MIB \subset OB$ clearly implies the channels are not more capable
comparable. This is because when the channels are more capable comparable we
know from \cite{elg79} that superposition coding is optimal and that $MIB = CR =
OB$. \end{proof}

\end{document}